\documentclass[onecolumn,notitlepage,nofootinbib,floatfix,superscriptaddress]{revtex4-2}
\usepackage{amsmath, amsfonts, amsthm, amssymb, bbm}
\usepackage[margin=1in]{geometry}
\usepackage{xspace}
\usepackage{algorithm,algpseudocode}
\usepackage[endLComment=~]{algpseudocodex}
\usepackage{graphicx,xcolor,tikz}
\usepackage{hyperref}
\usepackage{physics}
\usepackage{qcircuit}

\newcommand{\bb}{\mathbb}
\newcommand{\mc}{\mathcal}
\newcommand{\res}[2]{\left. #1 \right|_{#2}}

\DeclareMathOperator{\supp}{supp}

\DeclareMathOperator{\poly}{poly}
\DeclareMathOperator*{\E}{\mathbb{E}}

\algblockdefx[ForEach]{ForEach}{EndForEach}[1]{\textbf{for each} #1 \textbf{do}}{\textbf{end for each}}

\newtheorem{theorem}{Theorem}
\newtheorem{lemma}[theorem]{Lemma}

\newtheorem{definition}[theorem]{Definition}
\newtheorem{property}[theorem]{Property}

\begin{document}

\title{A polynomial-time approximation scheme\\ for minimum-weight decoding of topological codes}

\author{Shouzhen Gu}
\thanks{These authors contributed equally to this work.}
\affiliation{Yale Quantum Institute \& Department of Applied Physics, Yale University, New Haven, CT, USA}
\author{Lily Wang}
\thanks{These authors contributed equally to this work.}
\affiliation{CSE Division, University of Michigan, Ann Arbor, MI, USA}
\author{Aleksander Kubica}
\affiliation{Yale Quantum Institute \& Department of Applied Physics, Yale University, New Haven, CT, USA}

\begin{abstract}
Two-dimensional topological translationally invariant (2D TTI) stabilizer codes lie at the heart of fault-tolerant quantum computation, but using them requires solving the decoding problem.
Minimum-weight decoding of these codes was recently shown to be NP-hard, even in basic settings, such as the color code with Pauli $Z$ errors and the toric code with Pauli $X$, $Y$ and $Z$ errors.
Here, we prove that minimum-weight decoding of 2D TTI codes nonetheless admits a polynomial-time approximation scheme (PTAS), i.e., for any constant $\varepsilon>0$, a recovery operator of weight within a multiplicative factor of $1+\varepsilon$ of the minimum can be found in polynomial time.
Our approach builds on Arora's PTAS for Euclidean problems, such as the traveling salesman problem, and applies when decoding can be cast in terms of point-like excitations connected by string-like errors.
It therefore extends beyond two dimensions, covering certain higher-dimensional topological codes and quantum memories, including the toric code with phenomenological or circuit-level noise.

\end{abstract}

\maketitle

\section{Introduction}
\label{sec:intro}

Quantum error-correcting (QEC) codes~\cite{Shor1995,Steane1996,Gottesman1996} are an indispensable component of fault-tolerant quantum computation~\cite{Shor1996,Preskill1998}.
Topological stabilizer codes~\cite{Kitaev2003, bombin2013introductiontopologicalquantumcodes, Kubica_thesis} constitute an important class of QEC codes, as they typically exhibit high QEC thresholds and admit fault-tolerant implementation of logical gates with low resource overhead.
The canonical example of a two-dimensional topological translationally invariant (2D~TTI) code is the toric code~\cite{Kitaev2003}, together with its planar variant, the surface code~\cite{bravyi1998}, which can be realized by placing qubits on a square lattice and measuring geometrically local parity checks.
Recently, much effort has been devoted to improving the encoding rate of the toric code, leading to the discovery of other 2D TTI codes, such as bivariate bicycle codes\footnote{To view bivariate bicycle codes as 2D TTI codes, we fix a set of polynomials in their polynomial representation~\cite{Haah2013} and consider a code family on lattices of growing size.}~\cite{Kovalev2013,bravyi2024high} and tile codes~\cite{steffan2025tile}.

For QEC codes to be practically relevant for fault-tolerant quantum computation, they must also admit computationally efficient solutions to the decoding problem~\cite{Terhal2015}.
The decoding problem asks the following question: given an error syndrome indicating the presence of errors, find a recovery operator that mitigates the detrimental effects of these errors on the encoded information.
For topological codes, many strategies have been devised based on minimum-weight decoding, in which one seeks a minimum-weight recovery operator consistent with the given syndrome.
Such strategies are well suited to scenarios with independent and identically distributed noise, where lower-weight errors are more likely, and they exhibit competitive (albeit not necessarily optimal) performance.
The prototypical example is the minimum-weight perfect matching decoder~\cite{Dennis2002} for the toric code, which, for either Pauli $X$ or $Z$ errors, finds a minimum-weight recovery operator;
other examples include matching decoders for the color code~\cite{Delfosse2014,Kubica2023CCrestrictiondecoder} and more general 2D TTI codes~\cite{tan2026generalizedmatchingdecoders2d,sahay2026matchingdecoderbivariatebicycle}. 

Despite the existence of many computationally efficient matching decoders, recent works~\cite{walters2026CCNPhard,gu2026colorcodesurfacecode} demonstrated that, rather surprisingly, minimum-weight decoding is NP-hard in basic QEC settings, including the toric code with Pauli $X$, $Y$ and $Z$ errors and the color code with Pauli $Z$ errors.
These results are interesting from a theoretical standpoint, as they establish hardness of the minimum-weight decoding problem for 2D TTI codes; from the practical standpoint, however, one may care more about the complexity of solving this problem approximately.
In that case, instead of finding a minimum-weight recovery, it suffices to find a recovery operator whose weight 
exceeds the minimum by at most a multiplicative factor of $1+\varepsilon$ for any constant $\varepsilon>0$.

In this work, we prove that minimum-weight decoding of 2D TTI codes admits a polynomial-time approximation scheme (PTAS), thereby circumventing the NP-hardness of the exact problem in the aforementioned basic QEC settings.
Specifically, to achieve approximation ratio $1+\varepsilon$, our algorithm has time complexity $L^2(\log L)^{O(1/\varepsilon)}$ for a code on an $L\times L$ lattice.
To our knowledge, our results constitute the first PTAS for minimum-weight decoding of 2D TTI codes; the previous best provable approximation ratio for the toric code was two, based on independent matching of $X$ and $Z$ excitations, with certain heuristic decoders trying to leverage the correlations between Pauli $X$ and $Z$ errors for the depolarizing noise~\cite{Delfosse2014XZ,Higgott2023,iOlius2023}.

Our approach is based on Arora's PTAS for Euclidean problems such as the traveling salesman problem~\cite{AroraTSP}.
The core idea is to recursively dissect the lattice into square regions and show that the optimal solution can be perturbed, with only a slight increase in cost, so that within each region the solution is locally optimal and crosses region boundaries only a constant number of times, and only through regularly spaced portals; such a solution can then be found efficiently by dynamic programming over the recursive dissection tree.
Fig.~\ref{fig:main} illustrates the main concepts used in the proof.
We explain the details of the approach first for the simpler case of the toric code before proving the result in the general setting of 2D TTI codes.
The key structure of the decoding problem that is leveraged is the presence of point-like excitations connected by string-like errors, which allows for generalizations to certain higher-dimensional topological codes and quantum memories, including 2D TTI codes with phenomenological or circuit-level noise.
The proof of correctness in these higher-dimensional decoding problems requires a technical condition on the structure of excitations; for concreteness, we verify this condition for the quantum memory setting with the toric code and phenomenological noise.
We expect our results to be not only of theoretical value, but also a starting point for practical decoders that improve on existing approaches seeking low-weight recovery operators, such as the minimum-weight perfect matching decoder.

\section{Preliminaries}
\label{sec:prelim}
An $n$-qubit stabilizer code $\mc C\subseteq (\bb C^2)^{\otimes n}$ is the simultaneous $+1$-eigenspace of an abelian subgroup $\mc S$ of the $n$-qubit Pauli group $\mc P^n$, where $-I\notin \mc S$.
By ignoring phases, we typically work with the symplectic representation of Pauli operators by identifying
them with binary vectors of length $2n$,
where the first $n$ components give the $X$ support and the last $n$ components give the $Z$ support. In this representation, multiplication of operators corresponds to addition (modulo two) of vectors.

Suppose $\mc S$ is generated by elements $s_1, \dots, s_p$. Then, if a Pauli error $e$ is applied to the code, the stabilizers values $m_i=(-1)^{\sigma_i}$ of the generators $s_i$ that anticommute with $e$ will be flipped from $+1$ to $-1$. The binary vector $\sigma\in \bb F_2^p$ is called the \emph{syndrome} of $e$.
For us, errors are called \emph{equivalent} if they produce the same syndrome, regardless of whether they differ by a logical operator. We refer to individual stabilizer generators that have been flipped, i.e., $\sigma_i=1$, as locations of \emph{excitations}. The task of decoding is to find an error $e'$ resulting in a given syndrome $\sigma$.
We always assume that the syndrome is valid, i.e., that there is some solution to the decoding problem; checking if this is the case can be done using Gaussian elimination in $O(n^3)$ time.
In the minimum-weight decoding problem, our goal is to find such an error $e_0$ of minimum weight, where the weight of a Pauli operator is defined to be the size of its support.
Because finding the minimum-weight error is often hard~\cite{Hsieh2011,Kuo2020,walters2026CCNPhard,gu2026colorcodesurfacecode}, we are interested in $(1+\varepsilon)$-approximations, meaning an error $e'$ satisfying the given syndrome and of weight $|e'|\le (1+\varepsilon)|e_0|$. If for some code family $\{\mc C_n\}_n$ and any fixed $\varepsilon>0$, an algorithm with time complexity $\poly(n)$ can find a $(1+\varepsilon)$-approximation for any syndrome, we say that it is a polynomial-time approximation scheme (PTAS) for the minimum-weight decoding problem. The dependence on $\varepsilon$ is arbitrary; for example, $O(n^{\exp(1/\varepsilon)})$ is allowed.
We often leave the subscript $n$ implicit when referring to a code family.

\begin{figure}[htp!]
    \centering
\includegraphics[width=0.75\linewidth,trim={2cm 3cm 3cm 3cm},clip]{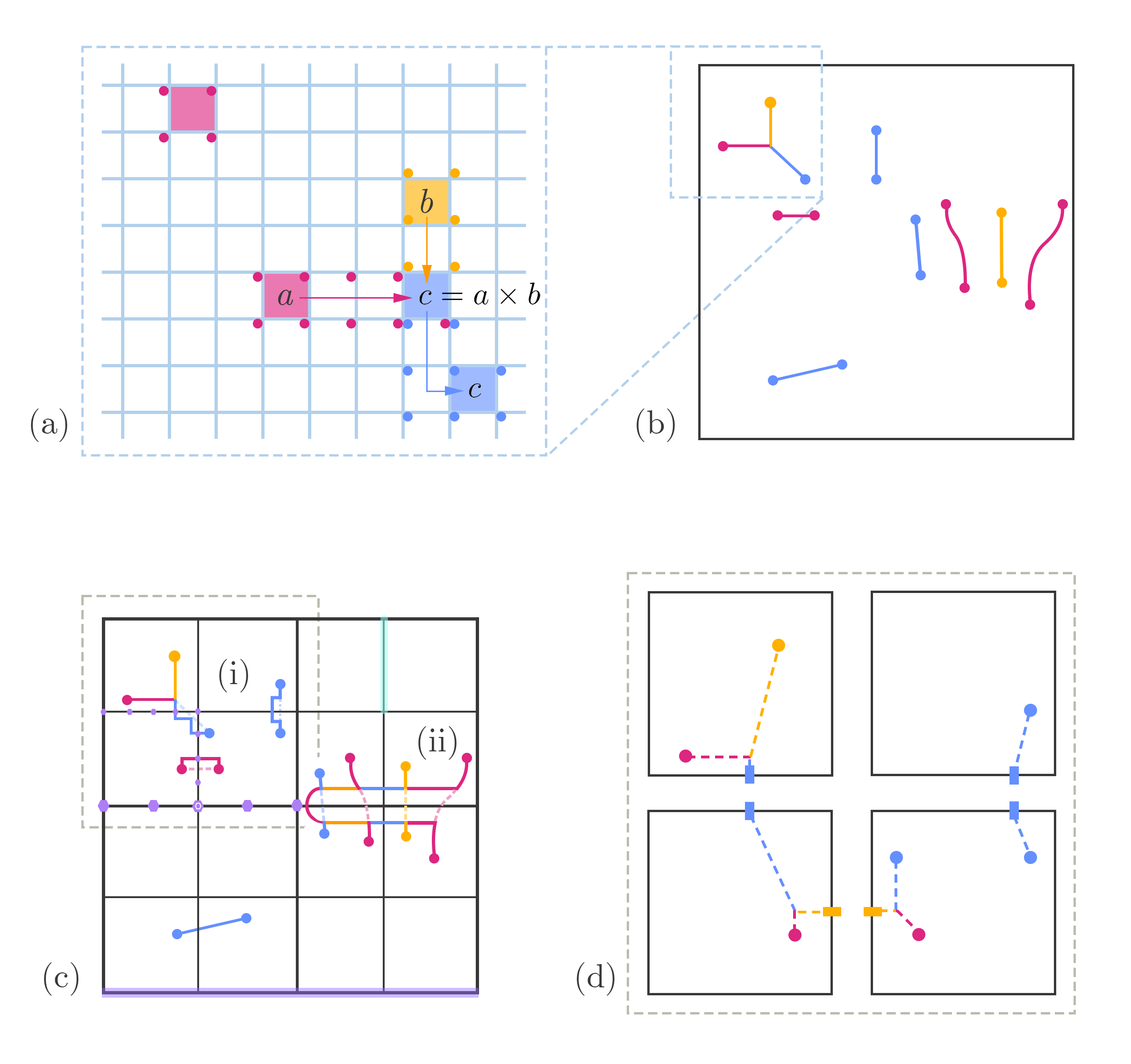}
    \caption{
    (a) A 2D TTI code on a square lattice has a constant number of qubits at each vertex; its stabilizers are associated with faces and supported only on adjacent qubits.
    Violated stabilizers correspond to abelian anyons, which are point-like excitations that are deterministically moved by string-like operators.
    Two anyons $a$ and $b$ can be fused together into $c = a\times b$.
    (b) In the decoding problem, for a given syndrome $\sigma$, i.e., a set of excitations (colored dots), we seek a corresponding error $e$ (colored lines) that has minimum or close to minimum weight.
    (c) To solve it, we recursively partition the lattice into squares; we depict level-$1$ (thick) and level-$2$ (thin) squares.
    The level of a line or line segment is the level of the largest square coinciding with it.
    For example, the turquoise line segment is level 2 and the purple line is level 1.
    Portals (purple hexagons) are regularly spaced along the sides of each square; the number of portals for each square is the same, independent of its level.
    A minimum-weight error $e_0$ error may be deformed using the (i) Rerouting and (ii) Patching Properties to obtain an equivalent error $e'$ of weight $|e'|\le (1+\varepsilon)|e_0|$ which crosses each line segment only at portals and uses at most $r$ of them.
    (d) The minimum-weight $r$-light error $e$ with syndrome $\sigma$ is found by dynamic programming;
    we depict squares within a region bounded by the gray dashed lines in (c).
    Solutions of decoding subproblems from smaller squares are combined to give a minimum-weight error on a larger square that uses at most $r$ portals;
    the global solution is guaranteed to have weight $|e|\le |e'|$.
    }
    \label{fig:main}
\end{figure}

In this work, we are primarily concerned with decoding 2D TTI codes.
A 2D TTI code $\mc C$ is defined on an $L\times L$ square lattice with periodic boundary conditions, although our results also hold with open boundary conditions, as we discuss in Sec.~\ref{subsec:removeassumptions}.
We place $q$ qubits on each site of the lattice, where $q$ is a constant. The stabilizers generators are the translations of a constant number of geometrically local operators, meaning that they only act on qubits within a constant radius; we then consider a code family on lattices of growing size.
Without loss of generality, we can associate each stabilizer generator with a plaquette and assume that it acts only on qubits on the vertices of that plaquette; see Fig.~\ref{fig:main}(a).
This is because we can coarse-grain the original lattice, grouping together qubits in a constant-sized unit cell as a single vertex, which results in a larger but constant number of qubits on each site.
We also require that $\mc C$ has topological order, meaning that its code distance grows with $L$.
A 2D TTI code is equivalent via a constant-depth local Clifford circuit to some number $\alpha$ of copies of the toric code (and qubits in a product state)~\cite{Yoshida2011,Haah2013,bombin2014structure}.
Therefore, we can associate a charge in $\bb Z_2^{2\alpha}$ to any excitation in $\mc C$, which is sometimes called an abelian anyon. An excitation may be moved around via a string-like operator to a different plaquette on the lattice while maintaining its charge. Two excitations $a$ and $b$ that are moved to the same location may combine into a single excitation $c$, which is conventionally denoted by a fusion rule $c=a\times b$.
The charge of $c$ is found by adding modulo two the charges of $a$ and $b$.
By coarse-graining sufficiently, such an operator can be made to have support only on qubits on the vertices adjacent to the path taken (on the dual lattice). In particular, its weight is at most $(4+2s)q$ for a path of $\ell^1$-length $s$.
We say that the total charge of a set of anyons is neutral if the charges sum to zero. This occurs only if the excitations form the syndrome of some error, which annihilates the excitations and can be taken to be supported on paths that move all the excitations to a given location. Conversely, the total charge of the excitations caused by any error is neutral.
Fig.~\ref{fig:main}(a) illustrates the properties of TTI codes described above.

An example of a 2D TTI code is the toric code.
In its usual representation, qubits are placed on edges of the lattice and $X$ and $Z$ stabilizers on the vertices and plaquettes, respectively, with supports on the neighboring qubits. To formulate the toric code in the framework of 2D TTI codes, we move the qubits down or to the left along its edge to the nearest vertex so that there are instead $q=2$ qubits on each vertex; see Fig.~\ref{fig:TTIsurfacecode}(a).
The plaquette $Z$ stabilizers stay at the same locations, whereas each vertex $X$ stabilizers is shifted to the neighboring plaquette at its lower left (Fig.~\ref{fig:TTIsurfacecode}(b)(c)). Thus, there is now one stabilizer of each of the two Pauli types placed on each plaquette, which touches a subset of qubits on the neighboring vertices.
Lastly, we refer to the locations of flipped $X$ and $Z$ stabilizers as $X$ and $Z$ excitations, respectively.

\section{Setup}
\label{sec:overview}
Our main result is to present a PTAS for the minimum-weight decoding problem for any 2D TTI code.
\begin{theorem}
    \label{thm:PTAS2DTTI}
    For any 2D TTI code $\mc C$ on an $L\times L$ lattice and fixed $\varepsilon>0$, there is a randomized algorithm with time complexity $L^2(\log L)^{O(1/\varepsilon)}$ that achieves approximation ratio $1+\varepsilon$ for the minimum-weight decoding problem and succeeds with probability at least $1/4$. A derandomized version of the algorithm, which always succeeds, has the same space complexity $L^2(\log L)^{O(1/\varepsilon)}$ and runs in time $L^4(\log L)^{O(1/\varepsilon)}$.
\end{theorem}

Although we present the randomized algorithm with success probability 1/4, its parameters can be easily changed to achieve any constant success probability less than one.
Alternatively, we can run the algorithm multiple times and keep the minimum-weight solution over all runs to increase the success probability.

The $L\times L$ lattice on which the code is defined has vertices with coordinates $(x,y)\in \bb Z_L^2$.
For simplicity, we will consider the case where $L$ is a power of two, although the general setting can be handled with minor modifications, as we describe in Sec.~\ref{subsec:removeassumptions}.
We define a \emph{dissection} of the lattice into squares as illustrated in Fig.~\ref{fig:main}(c).
The lattice, considered a level-$0$ square, is first divided into four squares of side length $L/2$ at locations $x=0, L/2$, $y=0, L/2$, which we call level-$1$ squares. For $i\ge 1$, each level-$i$ square is recursively partitioned into four smaller squares which have level $i+1$ and side length $L/2^{i+1}$.
This process repeats until we arrive at the squares with side length $s_0$. The smallest side length $s_0$ will not depend on $L$ but may depend on the decoding accuracy required. The maximum level is therefore $i_0=\log_2(L/s_0)$.

Relative to the placement of the squares described above, we consider shifts. Concretely, an $(a,b,c,d)$-\emph{shifted dissection} is a translation of the dissection by $(as_0+c, bs_0+d)$, where $a,b,c,d\in \bb Z$ and $0\le a,b < L/s_0$, $0\le c,d < s_0$.
Most of our analysis holds for any fixed $(c,d)$, and we often leave it implicit when this is the case.
Only in Lemma~\ref{lem:randomshiftmovesyndrome} of Sec.~\ref{sec:generalTTI} do we randomize over $(c,d)$.

It will also be useful to consider vertical and horizontal lines along the lattice at distance $s_0$ apart which align with the sides of the squares in the shifted dissection, i.e., at coordinates $x$ where $s_0\mid (x-c)$ or $y$ where $s_0\mid (y-d)$.
Each line contains many line segments which are the sides of the squares.
The level of a line segment is the level of the square that it is a side of, and the level of a line with respect to a shifted dissection is the level of the largest line segments it contains; see Fig.~\ref{fig:main}(c). A level-$i$ line contains $2^i$ level-$i$ line segments.
Note that the locations of valid lines depend on $(c,d)$. A line is at a fixed location, but its level and the line segments contained in it depend on the shift $(a,b)$.\footnote{A vertical line with $x$-coordinate $x$ has level $i$ which is the smallest positive integer such that $\frac{L}{2^i}\mid (x-as_0-c)$. Similarly, a horizontal line with $y$-coordinate $y$ has level $i$ which is the smallest positive integer such that $\frac{L}{2^i}\mid (y-bs_0-d)$.}

Syndromes and errors can be decomposed with respect to a square $S$ in a shifted dissection. Recall that each stabilizer generator is placed on a plaquette of the lattice. We can therefore express any syndrome $\sigma$ as a disjoint union $\sigma = \sigma_S\sqcup \sigma_{S^c}$, where $\sigma_S$ and $\sigma_{S^c}$ are the restrictions of $\sigma$ to the stabilizers \emph{inside} and \emph{outside} of $S$, respectively. Here, we identify a binary vector with its set of nonzero components. Qubits, which support errors, are on vertices which are either in the interior of $S$, in the exterior of $S$, or on the \emph{boundary} of $S$ (denoted $\partial S$) comprising its line segments.
The level-$0$ square has no exterior or boundary due to periodic boundary conditions. Because the stabilizer generators only act on qubits on the adjacent vertices, errors on qubits in the interior or exterior of $S$ only cause excitations on plaquettes inside and outside $S$, respectively.
Errors on $\partial S$ may cause excitations on plaquettes both inside and outside $S$. The boundary of a square can be further decomposed into \emph{interiors of line segments}, which exclude their endpoint vertices, and \emph{corners}. An error on a qubit in the interior of a level-$i$ line segment can only cause excitations inside the level-$i$ squares on either side of the line segment, whereas one on a corner qubit may cause excitations inside up to four level-$i$ squares---the ones containing that corner.
There are at most $4q$ qubits on the corners of any square.
Note that a qubit in the interior of a line segment may be on a corner of a smaller square. An error on a line may cause excitations on either side of the line.

For every square $S$, we will identify a subset of the qubits $P_S$ on $\partial S$ as \emph{portals} (see Fig.~\ref{fig:main}(c)). If $T$ is a side of $S$, we may also write $P_T$ to denote the portals of $S$ on $T$. The portals of a square are the qubits supported on the following vertices in $\partial S$: all corners of $S$, $m'$ vertices equally spaced on each side of $S$ (or all vertices if $m'$ is greater than the side length), all vertices that are one away from any of the previously mentioned vertices, and the eight vertices that are distance two away from the corners. If a square has side length less than $2m'$, then every vertex is included.
The number of portals on the interior of each line segment is at most $m=(3m'-2)q$, which is fixed for all squares; the portals on the larger squares are spaced farther apart.
We will choose $m=O(\frac{1}{\varepsilon}\log L)$ in our proof. From our definition, if $S_1$, $S_2$, $S_3$, $S_4$ are the largest proper subsquares of a square $S$, then each portal of $S$ is the portal of some subsquare, i.e., $P_S\subseteq \bigcup_{i=1}^4 P_{S_i}$.

\section{Dynamic Program}
Consider a 2D TTI code $\mc C$. We wish to find a $(1+\varepsilon)$-approximation to the minimum-weight decoding problem for a given syndrome $\sigma$. After dissecting the lattice into squares and placing portals on their boundaries as described in Sec.~\ref{sec:overview}, we find a solution that does not use too many portals.

\begin{definition}
    An error is $r$-light with respect to a choice of portals defined by $m$ and a shifted dissection $(a,b,c,d)$ if its support on
    any line segment $T$ is contained in the portals $P_T$ and its weight on the interior of $T$ is at most $r$.
\end{definition}

In particular, an $r$-light error intersects the boundary $\partial S$ of any square $S$ at its portals $P_S$ and at most $r$ times on any side of $S$.
Also note that there is no restriction on the support of an $r$-light error on the $4q$ corner portals of $S$ and any such support does not contribute to the weight on the interiors of the sides.

\begin{figure}[htpb]
    \centering
    \includegraphics[width=0.85\linewidth,trim={4cm 4cm 4cm 4cm},clip]{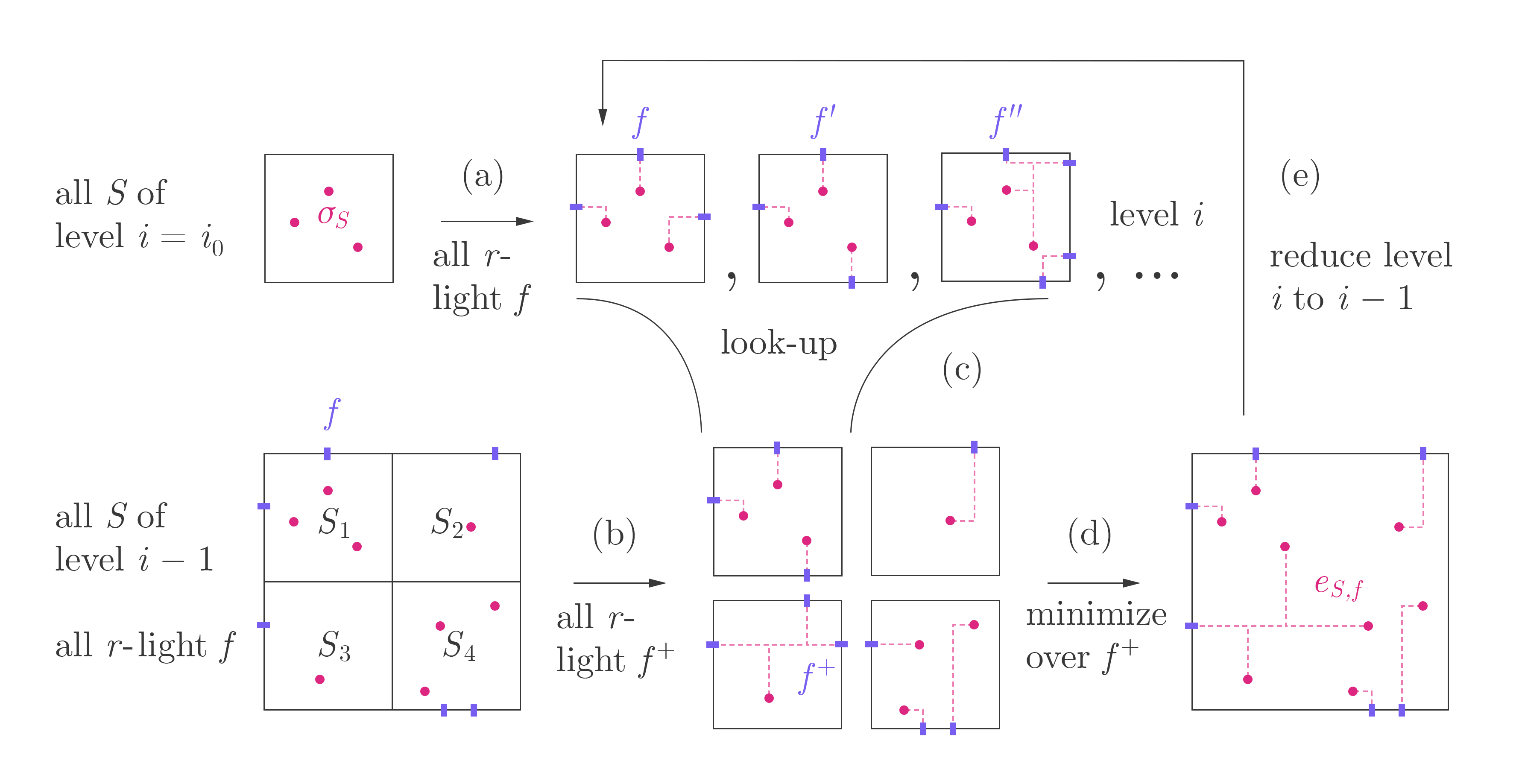}
    \caption{
    Finding the minimum-weight $r$-light error using dynamic programming.
    (a) In the base cases, for every level-$i_0$ square $S$, we consider all $r$-light $f$ supported on the portals $P_S$ and exhaustively search for the solution of the decoding subproblem.
    (b) To solve the decoding subproblems for a bigger square $S$ of level $i-1$ and an $r$-light $f$ on $P_S$, we consider all $r$-light lifts $f^+$.
    (c) The $r$-light lift $f^+$ defines decoding subproblems on the level-$i$ subsquares of $S$, whose solutions we look up.
    (d) Minimizing over $f^+$ gives the solution for the decoding subproblem defined by $S$ and $f$, (e) which we store in the lookup table.
    }
    \label{fig:DPoverview}
\end{figure}

In this section, we present Algorithm~\ref{alg:DP} (illustrated in Fig.~\ref{fig:DPoverview}) that finds the minimum-weight $r$-light error for a given syndrome $\sigma$.
When there are at most $m=O(\frac{1}{\varepsilon}\log L)$ portals on each side of every square and $r=O(\frac{1}{\varepsilon})$, our algorithm has time complexity $L^2(\log L)^{O(1/\varepsilon)}$. The analysis holds for any shift $(a,b,c,d)$, which we leave implicit.
In Secs.~\ref{sec:structuretheorem} and~\ref{sec:generalTTI}, we will show that the optimal $r$-light solution is likely a good approximation of the minimum-weight solution over a random choice of $(a,b,c,d)$ (see Theorem~\ref{thm:structure} and Lemma~\ref{lem:randomshiftmovesyndrome}).

\begin{algorithm}[H]
	\caption{Dynamic program}
	\label{alg:DP}
	\textbf{Input:} 2D TTI code $\mc C$; shifted dissection with maximum level $i_0$, squares $\{S\}$, and portals $\{P_S\}$; syndrome $\sigma$\\
	\textbf{Output:} The minimum-weight $r$-light error with syndrome $\sigma$
	\begin{algorithmic}[1]
        \LComment{Base cases; Fig.~\ref{fig:DPoverview}(a)}
        \ForAll{$S\gets$ level-$i_0$ square}
            \ForAll{$f\gets$ $r$-light error supported on $P_S$}
                \State $\mu_{f,S} \gets$ excitations caused by $f$ inside $S$
                \State $e_{S,f}\gets$ solution to the decoding subproblem with input $S$ and $\sigma_S+\mu_{f,S}$ if it exists\Comment{exhaustive search}
            \EndFor
        \EndFor
        \LComment{Recursive steps; Fig.~\ref{fig:DPoverview}(b)-(e)}
        \For{$i=i_0-1, \dots, 0$}
            \ForAll{$S\gets$ level-$i$ square}
                \ForAll{$f\gets$ $r$-light error supported on $P_S$}
                    \State $S_1, S_2, S_3, S_4 \gets$ level-$(i+1)$ subsquares of $S$
                    \ForAll{$f^+\gets$ $r$-light lift of $f$ to $\bigcup_{i=1}^4 P_{S_i}$}
                        \For{i=1, 2, 3, 4}
                            \State $f_i^+\gets$ restriction of $f^+$ to ${P_{S_i}}$
                            \State $e_{S_i,f_i^+}\gets$ solution of the decoding subproblem with input $S_i$ and $\sigma_{S_i}+\mu_{f_i^+,S_i}$ \Comment{lookup table}
                        \EndFor
                        \If{all $e_{S_i,f_i^+}$ exist}
                            \State $e_{S,f}^{f^+}\gets \res{f^+}{\bigcup_{i=1}^4 P_{S_i}\setminus \partial S} + \sum_{i=1}^4 e_{S_i,f_i^+}$
                        \EndIf
                    \EndFor
                    \State $e_{S,f}\gets$ $e_{S,f}^{f^+}$ with minimum weight if at least one exists
                \EndFor
            \EndFor
        \EndFor
        \State $e\gets e_{S,I}$ where $S$ is the level-$0$ square
        \State \Return $e$
	\end{algorithmic}
\end{algorithm}

The algorithm uses dynamic programming to break down the decoding problem on the entire code into smaller subproblems on squares of the code. This approach is reminiscent of renormalization-group-type decoders since we coarse-grain the lattice into self-similar regions~\cite{HarringtonThesis,DuclosCianci2010,BravyiHaah}.
For us, each subproblem will be defined by the original excitations inside the square as well as ones that may have been moved inside $S$ through the portals of $S$ by a fixed set of errors supported there.
We must find the minimum-weight $r$-light solution in the interior of $S$ consistent with this syndrome. Note that we may not move excitations outside of $S$ through the boundary in this solution. However, by iterating through the possible $r$-light errors supported on the portals of smaller subsquares and piecing together their optimal solutions, we are able to solve the subproblem for a larger square.

\begin{definition}
    A decoding subproblem is specified by a square $S$ and a set of excitations inside $S$ as input, and the task is to find a minimum-weight $r$-light solution contained in the interior of $S$ or determine that no valid solution exists.
\end{definition}

We now describe the algorithm. The base cases of the dynamic program consist of the following decoding subproblems (Fig.~\ref{fig:DPoverview}(a)).
Let $S$ be a level-$i_0$ square and $\sigma_S$ be the restriction of the global syndrome $\sigma$ to the stabilizers inside $S$.
Let $f$ be an $r$-light error supported on the portals of $S$ and $\mu_{f,S}$ be its syndrome restricted to the inside of $S$. For every such $S$ and $f$, we solve the decoding subproblem given by $S$ and $\sigma_S + \mu_{f,S}$ by exhaustive search. The solution $e_{S,f}$, or the fact that none exists, is stored in a lookup table.
If there are multiple solutions of the same minimal weight, only one is stored.

Now assume that we have the optimal solution to every decoding subproblem given by a square of level greater than $i$ and the set of excitations arising from an $r$-light error supported on the portals of the square. Consider a decoding subproblem defined by a level-$i$ square $S$ and an arbitrary $r$-light error $f$ supported on its portals. (Note that for any square, the number of portals does not depend on its level.)
We wish to find its solution $e_{S,f}$. Let $S_1$, $S_2$, $S_3$, $S_4$ be the four subsquares of $S$ at level $i+1$.
To use the solutions from the decoding subproblems of the subsquares, we need to consider $r$-light errors $f^+$ supported on the portals of the subsquares $\bigcup_{i=1}^4 P_{S_i}$ such that $\res{f^+}{\partial S} = f$ (Fig.~\ref{fig:DPoverview}(b)). For every $r$-light error $f^+$, which we call an $r$-light lift of $f$, we look up the solutions $e_{S_i,f^+_i}$ to the decoding subproblems specified by $S_i$ and $\sigma_{S_i}+\mu_{f^+_i,S_i}$ for $i\in\{1,2,3,4\}$, where $f^+_i$ is the restriction of $f^+$ to $P_{S_i}$ (Fig.~\ref{fig:DPoverview}(c)). If valid solutions exist for all $i$, this gives a candidate solution
\begin{equation}
    e_{S,f}^{f^+} = \res{f^+}{\bigcup_{i=1}^4 P_{S_i}\setminus \partial S} + \sum_{i=1}^4 e_{S_i,f_i^+}.
\end{equation}
We then store the minimum-weight $e_{S,f}^{f^+}$ over all $f^+$ that have valid solutions for all $i$ in the lookup table as the solution $e_{S,f}$ (or that no solution exists if no such $f^+$ exists), as in Fig.~\ref{fig:DPoverview}(d).

In the final iteration, we solve the decoding subproblem for the level-$0$ square $S$, which is the entire lattice. That is, we take $S_1$, $S_2$, $S_3$, $S_4$ to be the four level-$1$ squares, $f=I$ since $S$ has no boundary, and output $e$ as the minimum-weight $e_{S,f}^{f^+}$ over all valid $f^+$.

\begin{theorem}
    \label{thm:DP}
    Algorithm~\ref{alg:DP} finds the minimum-weight $r$-light error for the given syndrome $\sigma$. If $m=O(\frac{1}{\varepsilon}\log L)$ and $r=O(\frac{1}{\varepsilon})$, it has time complexity $L^2(\log L)^{O(1/\varepsilon)}$.
\end{theorem}

\begin{proof}
    We first show correctness of the algorithm. In particular, we claim that each $e_{S,f}$ is the solution to the decoding subproblem with input $S$ and $\sigma_S + \mu_{f,S}$. The base cases for level-$i_0$ squares is clear. Now consider a lower-level square $S$ and $r$-light error $f$ supported on $P_S$. All the candidate solutions $e_{S,f}^{f^+}$ are $r$-light by construction since $f^+$ and each $e_{S_i,f_i^+}$ are $r$-light, and only one of them can overlap with the interior of any line segment.
    
    Let $\tilde e_{S,f}$ be the true solution of the decoding subproblem. Since $\tilde e_{S,f}$ and $f$ are $r$-light and disjoint, and there is no interior of any line segment that intersects them both, their union is also $r$-light. Therefore, $\res{(\tilde e_{S,f} + f)}{\bigcup_{i=1}^4 P_{S_i}}$ is one of the $r$-light lifts $f^+$ that is considered during the algorithm.
    A minimum-weight $r$-light solution will be optimal in the interior of each subsquare. Indeed, any combination of $r$-light errors for the subsquares can be combined with $f^+$ on their boundaries to give an $r$-light error on $S$, so the minimum-weight ones should be chosen. Therefore, the solution $e_{S,f}$ pieced together from those in the lookup table $e_{S_i,f_i^+}$ will have the same optimal weight as $\tilde e_{S,f}$. The minimum-weight $r$-light error giving syndrome $\sigma$ is the decoding subproblem with input the level-$0$ square and $\sigma$, so the output $e$ of the algorithm is correct.    
    
    Now let us analyze the time complexity of the algorithm. There are $4^i$ squares of any level $i$, giving a total of $\sum_{i=0}^{i_0}4^i=O(L^2)$ squares. For each square, the number of $r$-light errors supported on its portals is at most
    \begin{equation}
        4^{4q}\left(\sum_{r'=0}^r\binom{m}{r'}3^{r'}\right)^4 \le 4^{4q+4r}(r+1)^4 m^{4r},
    \end{equation}
    since at most $r$ of the $m$ portals on each side could support one of the three nontrivial Pauli operators and the support on the $4q$ corner portals is arbitrary. Therefore, the total number of subdecoding problems that needs to be solved is $O(L^24^{4q+4r}r^4m^{4r})$.
    Each base case subdecoding problem takes time $2^{O(s_0^2)}$ using exhaustive search, which is constant.
    Decoding subproblems at level $i<i_0$ are solved by iterating through $r$-light lifts, finding the corresponding solutions to subdecoding problems at level $i+1$, and piecing them together.
    Because there are four line segments that are shared between two of the four level-$(i+1)$ squares as well as one corner shared by all four subsquares, the number of $r$-light lifts, and hence the number of solutions to retrieve in the lookup table, is at most
    \begin{equation}
        4^{q}\left(\sum_{r'=0}^r\binom{m}{r'}3^{r'}\right)^4 \le 4^{q+4r}(r+1)^4 m^{4r}.
    \end{equation}
    Therefore, for $m=O(\frac{1}{\varepsilon}\log L)$ and $r=O(\frac{1}{\varepsilon})$, the total time complexity of the algorithm is 
    \begin{equation}
    O(L^24^{5q+8r}r^8m^{8r}) = L^2(\log L)^{O(1/\varepsilon)},
    \end{equation}
    as $\varepsilon$ and $q$ do not depend on $L$.
\end{proof}

    We also note that the space complexity can be bounded as
    \begin{equation}
        \sum_{i=0}^{i_0}4^i O(4^{4q+4r}r^4m^{4r}) O((L/2^i)^2) = O(L^2 4^{4q+4r} r^4m^{4r}i_0)=O(L^2(\log L) r^4m^{4r}),
    \end{equation}
    since there are $4^i O(4^{4q+4r}r^4m^{4r})$ decoding subproblems on level-$i$ squares, each of which has a solution of size the area of the square $O((L/2^i)^2)$. Again, this expression is $L^2(\log L)^{O(1/\varepsilon)}$ when $m=O(\frac{1}{\varepsilon}\log L)$ and $r=O(\frac{1}{\varepsilon})$, which achieves only the trivial bound as it is the same as the time complexity.

\section{Structure Theorem}
\label{sec:structuretheorem}

Having shown that optimal $r$-light errors can be found efficiently, we now argue that they will give a good approximate solution to the minimum-weight decoding problem.
The main goal of this section is to prove the Structure Theorem, which states that an error may be transformed to an $r$-light one without significantly increasing its weight.
Our approach mimics closely Arora's proof of PTAS for Euclidean problems~\cite{AroraTSP}. Namely, we prove the Structure Theorem by first reducing the number of intersections of the error with all line segments, a process called \emph{patching}, and then shifting the intersections to the portals, a process called \emph{rerouting}. 
We define two properties of a syndrome $\sigma$, called the Patching and Rerouting Properties, that capture whether these two steps are possible. Together, they imply the Structure Theorem.
Because the results of this section hold for any shift $(c,d)$, we leave those values implicit.

\begin{theorem}[Structure Theorem]
    \label{thm:structure}
    Let $\mc C$ be a 2D TTI code, $e_0$ be an error with syndrome $\sigma$, and $\varepsilon>0$.
    Consider a dissection with at most $m=c_1(\log L)/\varepsilon$ portals on each line segment, and let $r=c_2/\varepsilon$, where $c_1$ and $c_2$ are constants.
    If $\sigma$ satisfies the Patching and Rerouting Properties~\ref{lem:patching} and~\ref{lem:rerouting},
    then with probability at least 1/2 over a random shift $(a,b)$, there exists an $r$-light error $e$ with syndrome $\sigma$ and weight $|e|\le (1+\varepsilon)|e_0|$.
\end{theorem}

The error $e$ in Theorem~\ref{thm:structure} is obtained by replacing components of the original error $e_0$ on the interiors of line segments with equivalent operators. These equivalent operators will be supported on a slightly larger region, which we call the \emph{buffer} of a line segment. The buffer $B_T$ of a vertical (horizontal) line segment $T$ consists of the interior of $T$ as well as its horizontal (vertical) translations by up to two units. For the special case of the toric code, we may instead define the buffer to include translations by only one unit; see Fig.~\ref{fig:TTIsurfacecode}(d).
The following lemma characterizes the overlap between a line segment and the buffers of other line segments.
In particular, buffers of shorter line segments do not intersect longer line segments, and buffers of all longer line segments can only intersect a shorter line segment at a constant number of locations.

\begin{lemma}
    \label{lem:bufferproperties}
    Let $s_0\ge 3$. There is a constant $A$ such that for any level-$i$ line segment $T$ the following holds.
    \begin{enumerate}
        \item If $T'$ is a different line segment of level at least $i$, then $(B_{T'}\setminus T')\cap T=\emptyset$.
        \item $\left|\left(\bigcup_{T': \operatorname{level}(T')<i} (B_{T'}\setminus T')\right)\cap T\right| \le A$.
    \end{enumerate}
\end{lemma}
\begin{proof}
    If $T'$ has level at least $i$, then either it is entirely contained in $T$, it is disjoint from $T$, or it intersects $T$ at the endpoint of $T'$. In all cases, $(B_{T'}\setminus T')\cap T=\emptyset$. If $T'$ has level less than $i$, i.e., $\operatorname{level}(T')<i$, then an endpoint of $T$ may be part of the interior of $T'$. In this case, $B_{T'}\setminus T'$ intersects $T$ at the two vertices one and two units away from the endpoint. Thus, there are only four vertices that can intersect $B_{T'}\setminus T'$ for a lower-level $T'$, and we may take $A=4q$.
\end{proof}

\begin{figure}[htpb]
    \centering
\includegraphics[width=0.75\linewidth,trim={1cm 1cm 1cm 1cm},clip]{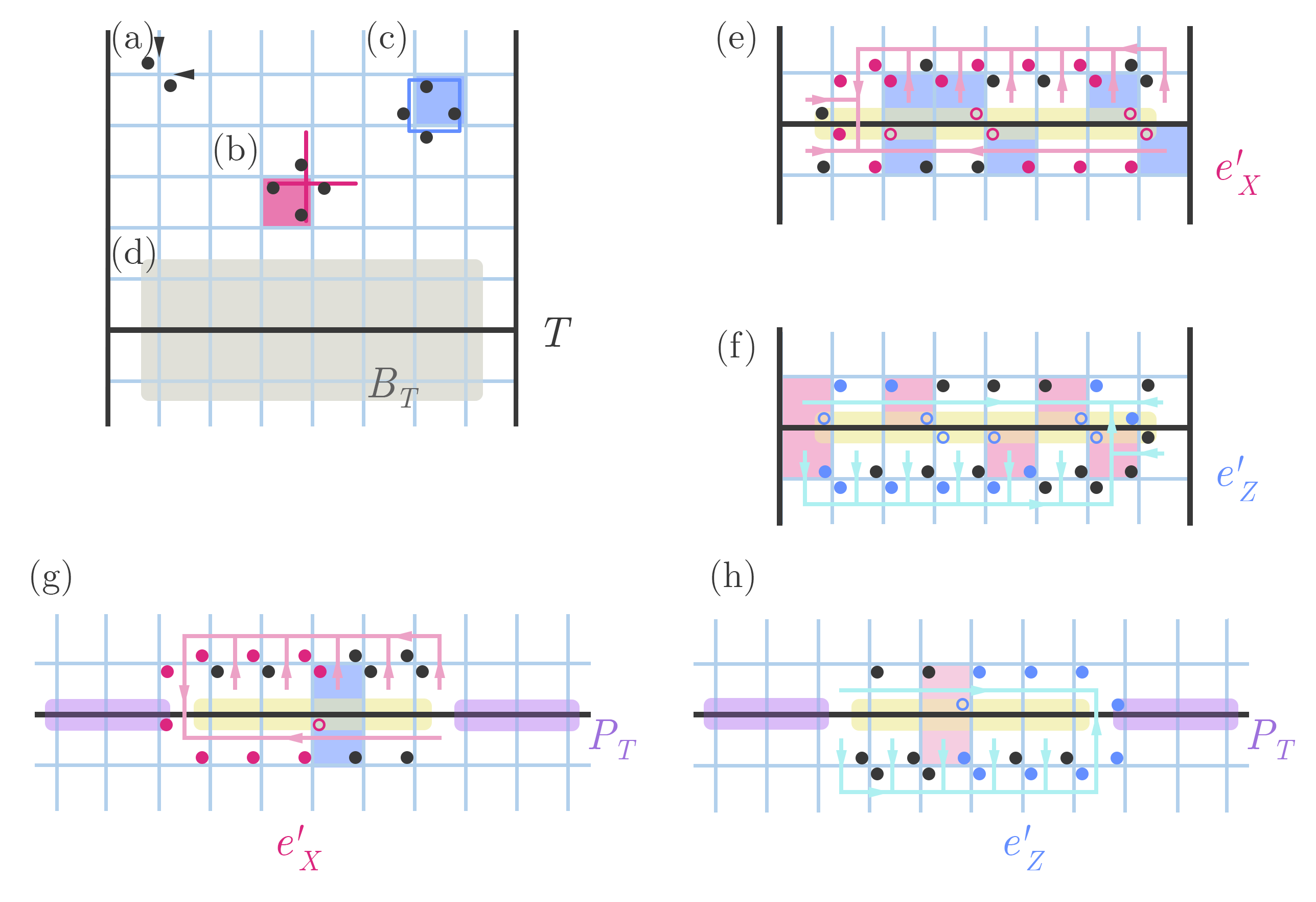}
    \caption{
    (a) To view the toric code in the framework of 2D TTI codes, we move qubits (black dots) from the edges of the lattice in its usual representation to the vertices; (b) and (c) depict $X$ and $Z$ stabilizers, respectively.
    (d) The buffer $B_T$ of a line segment $T$.
    (e) For an $X$ error $e_X$ (hollow red dots in the yellow region), the Patching Property is proven by moving the corresponding $Z$ excitations (blue squares) along the red path with a string-like operator $e_X'$ (red dots).
    (f) $X$ excitations created by a $Z$ error $e_Z$ (hollow blue dots) are moved analogously along the blue path with a string-like operator $e_Z'$.
    Then, $e'=e_X' + e_Z'$ is equivalent to $e = e_X + e_Z$.
    (g) In the proof of the Rerouting Property, $Z$ excitations (blue squares) created by an $X$ error (the hollow red dot) are moved along the red path using an $X$ operator (solid red dots), which only intersects $T$ on its portals.
    (h) $X$ excitations are rerouted analogously along the blue path.
    }
    \label{fig:TTIsurfacecode}
\end{figure}

Now, we are ready to state the Patching and Rerouting Properties.
For simplicity, we focus on the toric code case and prove that these properties hold for any syndrome.
We defer the analysis of any 2D TTI code with certain restrictions on $\sigma$ to Sec.~\ref{sec:generalTTI}.

\begin{property}[patching for syndrome $\sigma$]
    \label{lem:patching}
    There exist constants $g$, $\beta$, such that if $e_0$ is an error with syndrome $\sigma$ and $e$ is the restriction of $e_0$ to the interior of a line segment $T$ with length $s$, then there exists an error $e'$ supported on the buffer $B_T$ that is equivalent to $e$ and satisfies the following properties
    \begin{enumerate}
        \item $|e'| \leq gs $,
        \item $|\supp e' \cap T| \leq \beta$,
        \item for any other line segment $T'$ of level lower than or equal to the level of $T$, $\supp e' \cap T' = \emptyset$.
    \end{enumerate}
    
\end{property}

\begin{proof}[Proof (toric code)]
    The syndrome of $e$ consists of excitations whose total charge is neutral and which are located on plaquettes adjacent to $T$. The error $e'$ will be obtained from string operators that annihilate all excitations by moving all $X$ excitations to one place and all $Z$ excitations to one place.
    Without loss of generality, suppose $T$ is a horizontal line segment; a similar procedure works for vertical line segment because in the definition of the code, there is a diagonal symmetry switching $x$ and $y$ coordinates and the two qubits on every site.

    We first consider $Z$ excitations. We move the ones below $T$ horizontally to the plaquette second from the left. This only uses qubits in $B_T$ below $T$. We move the excitations above $T$ that are not the left-most one up by one unit, horizontally to the plaquette second from the left, and then down two units to fuse with the excitations below $T$. The left-most excitation above $T$, if present, is moved one unit right followed by one unit down to fuse with the rest of the excitations. These moves are supported on qubits in $B_T$, and their intersection with $T$ is at most the two qubits on the vertex second from the left; see Fig.~\ref{fig:TTIsurfacecode}(e).

    The $X$ excitations are moved in the opposite way (Fig.~\ref{fig:TTIsurfacecode}(f)). The ones above $T$ are moved horizontally to the plaquette second from the right. The ones below $T$ that are not the right-most one are moved one unit down, then horizontally to the plaquette second from the right, and then two units up. The right-most excitation is moved one unit left and one unit up. All of these moves are supported on $B_T$, and their intersection with $T$ is at most the two qubits on the vertex second to the right.

    Because $e'$ only has support on two vertices in $T$ and the vertices immediately above and below the interior of $T$, the first property is satisfied with $g=4$ and the second with $\beta=4$. The third property holds because $e'$ is supported on $B_T$ as a consequence of Lemma~\ref{lem:bufferproperties}.
\end{proof}

\begin{property}[rerouting for syndrome $\sigma$]
    \label{lem:rerouting}
    There exist constants $g'$, $\beta'$, such that if $e_0$ is an error with syndrome $\sigma$ and $e$ is the restriction of $e_0$ to the interior of a level-$i$ line segment $T$, then there exists an equivalent error $e'$ supported on the buffer $B_T$ satisfying the following properties
    \begin{enumerate}
        \item $\supp e'\cap T\subseteq P_T$,
        \item $|\supp e'\cap T|\le \beta' |e|$,
        \item $|e'|\le \frac{g'L}{2^im}|e|$.
    \end{enumerate}
\end{property}
\begin{proof}[Proof (toric code)]
    Again, without loss of generality, we may assume $T$ is a horizontal line segment.
    We apply the following procedure for all $T_k$ which is the portion of $T$ between two portals to eliminate the excitations caused by the restriction of $e$ to $T_k$.   
    Move all excitations below $T_k$ which are violations of $Z$ stabilizers horizontally to the plaquette left of $T_k$. Move all excitations from $Z$ stabilizers above $T_k$ up one unit, then horizontally to the plaquette left of $T_k$, then down two units to fuse with the excitations below $T_k$. Move all excitations above $T_k$ which are violations of $X$ stabilizers horizontally to the plaquette to the right of $T_k$. Move all excitations from $Z$ stabilizers below $T_k$ down one unit, then horizontally to the plaquette right of $T_k$, then up two units to fuse with the excitations above $T_k$. These moves are shown in Fig.~\ref{fig:TTIsurfacecode}(g)(h).
    The error $e'$ is the sum of the restriction of $e$ to $P_T$ and the moves from all $T_k$.

    All of the moves are supported on $B_T$ and their intersection with $T$ is contained in the portals, so the first property is satisfied. For each $T_k$, we use the portal on the left (right) only if $e$ has at least one $X$ ($Z$) error on a qubit in $T_k$ associated with a horizontal (vertical) edge in the usual representation of the toric code. Therefore, the second property is satisfied with $\beta'=1$.
    From our placement of portals, each $T_k$ has length at most $\frac{L}{2^i(m'-1)}-4$ because $T$ has length $L/2^i$, there are $m'$ equally spaced vertices including the endpoints of $T$, and portals at those locations and on either side. For each $T_k$ that contains an error, $e'$ may have support on the qubits above and below $T'$, the portals to the left and to the right of $T'$, and the vertices above and below those two portals.
    Thus, the number of vertices is $2(\frac{L}{2^i(m'-1)}-3)+6\le \frac{8qL}{2^im}$ since $m=(3m'-2)q\le 4(m'-1)q$, assuming $m'\ge 2$. Thus, the error weight may be increased by a factor of at most $\frac{g'L}{2^im}$ with $g'=8q^2=32$, satisfying the third property.
\end{proof}

For any line $\ell$ and error $e$, we define $t(\ell, e)=|\supp e\cap \ell|$ as the weight of $e$ restricted to $\ell$.
Since any qubit is in at most two lines, the following lemma holds.

\begin{lemma}\label{crossings-cost}
    For any error $e$, we have
    \begin{equation}
        \sum_{\mathrm{line\ }\ell} t(\ell, e) \leq 2 |e|.
    \end{equation}
\end{lemma}

We are now ready to prove the Structure Theorem.

\begin{proof}[Proof of Theorem~\ref{thm:structure}]
    We will show that for any fixed shift $(a,b)$, the error $e_0$ can be transformed in two steps into an equivalent $r$-light error $e$. Furthermore, when we randomize over the shift, the resulting error is likely to have weight $|e|\le (1+\varepsilon)|e_0|$. Let $g$, $\beta$, $g'$, $\beta'$ be as in Properties~\ref{lem:patching} and~\ref{lem:rerouting}.
    
    The first step is to modify the error $e_0$ to limit its support on each line segment by repeatedly applying the Patching Property~\ref{lem:patching}. Anytime the error has weight more than $(r-A)/\beta' - A$ on the interior of a line segment $T$, we replace its support on $T$ with an equivalent error supported on the buffer $B_T$ that has weight at most $\beta$ on $T$. This is done in decreasing order of line-segment levels, starting with the shortest, level-$i_0$ line segments and continuing until the longest, level-$1$ line segments have been patched. Let $e_1$ be the resulting error. Since we only add errors to the buffers of each line segment, Lemma~\ref{lem:bufferproperties} implies that the weight of $e_1$ on the interior of any line segment is at most $(r-A)/\beta'$. (We choose $r$ sufficiently large so that $\beta < (r-A)/\beta' - A$.)
    
    The second step is to apply the Rerouting Property~\ref{lem:rerouting} so that the intersections of the error with any line segment $T$ only occurs at the portals. We do this in increasing order of line-segment levels. Because portals of longer lines are also portals of shorter lines, we do not need to reroute on any line segment contained in a longer line segment. Let $e_2$ be the resulting error. The Rerouting Property~\ref{lem:rerouting} and Lemma~\ref{lem:bufferproperties} guarantee that the support of $e_2$ on any line segment is contained in its portals and its weight on the interior of the line segment is at most $r$. This shows that $e_2$ is an $r$-light error.
    
    We now analyze the expected increase in the weight of the error from these two steps, $|e_2|-|e_0|$, as we randomize over the shift $(a,b)$. Each application of the Patching or Rerouting Properties~\ref{lem:patching}~and~\ref{lem:rerouting} on a line segment will incur a \emph{cost}, which we associate with the line containing the line segment.
    For a vertical line $\ell$, let $X_{\ell,j}(b)$ be a random variable denoting the number of level $j$ line segments contained in $\ell$ that are patched, with the analysis similar for horizontal lines. Note that $X_{\ell,j}(b)=0$ if $j$ is smaller than the level of $\ell$. Also, $X_{\ell,j}(b)$ depends only on the vertical shift $b$ since that affects how the intersections of the error with $\ell$ are positioned relative to the line segments of different levels. As a level-$j$ line segment has length $L/2^j$, the cost of such a patch is bounded from above by $\frac{g L}{2^j}$. The total cost to patch $\ell$ is at most $\sum_{j=i}^{i_0} X_{\ell,j}(b) \frac{gL}{2^j}$, which depends on the level of $\ell$.
    For a random horizontal shift, the probability that $\ell$ is level $i$ is $2^{i-1-i_0}$ for $1<i\le i_0$ and $2^{1-i_0}$ for $i=1$, which is at most $2^{i-i_0}=2^i s_0/L$ in either case.
    Therefore, for any vertical shift $b$, the expectation over horizontal shifts $a$ of the cost associated with $\ell$ can be bounded as follows
    \begin{align}
        \E_{a}[\text{cost to patch }\ell] &= \sum_{i=1}^{i_0} \Pr(\operatorname{level}(\ell)=i) \sum_{j=i}^{i_0} X_{\ell,j}(b) \frac{gL}{2^j}
        \leq \sum_{i=1}^{i_0} \frac{2^{i} s_0}{L} \cdot \sum_{j=i}^{i_0}  X_{\ell,j}(b) \frac{gL}{2^j}\\
        &= s_0g\sum_{j=1}^{i_0} \frac{1}{2^j}X_{\ell,j}(b) \sum_{i=1}^{j} 2^i
        \le s_0g\sum_{j=1}^{i_0} \frac{1}{2^j}X_{\ell,j}(b) 2^{j+1}\\
        &= 2s_0g\sum_{j=1}^{i_0} X_{\ell,j}(b)
        \le \frac{2s_0gt(\ell,e_0)}{r'}.
        \label{eq_3}
    \end{align}
    The inequality in Eq.~\eqref{eq_3} uses the fact that each application of the Patching Property~\ref{lem:patching} decreases the support of the error on $\ell$ by at least $r':=(r-A)/\beta'-A-\beta$.
    By linearity of expectation and Lemma~\ref{crossings-cost}, the expected total cost of patching all lines is at most 
    \begin{equation}
        \label{eq:patchingcost}
        \sum_{\mathrm{line\ }\ell} \frac{2s_0gt(\ell,e_0)}{r'} \le \frac{4s_0g|e_0|}{r'}.
    \end{equation}

    Finally, we analyze the cost incurred by rerouting $e_1$. If $\ell$ is a level-$i$ line, an upper bound on the cost associated with $\ell$ can be obtained by considering the worst case in which all intersections of $e$ with $\ell$ occur on different level-$i$ line segments. Each such intersection will then incur a cost of at most $\frac{g'L}{2^i m}$. The expected cost for one given intersection is therefore
    \begin{align}
        \E_{(a,b)}[\text{cost of rerouting intersection}] &\le \sum_{i=1}^{i_0}\Pr(\operatorname{level}(\ell)=i)\frac{g'L}{2^i m}
        \le \sum_{i=1}^{i_0} \frac{2^i s_0}{L}\frac{g'L}{2^i m}\\
        &= \frac{s_0g'i_0}{m}
        = \frac{s_0g'\log_2(L/s_0)}{m}.
    \end{align}
    The total number of intersections of $e_1$ with all lines $\ell$ is at most $2|e_0|$ by Lemma~\ref{crossings-cost} and the fact that patching can only decrease the number of intersections. This means that the total expected rerouting cost is
    \begin{align}
        \label{eq:reroutingcost}
        \E_{(a,b)}[\text{total cost of rerouting}] \leq \frac{2s_0g'\log_2(L/s_0)|e_0|}{m}.
    \end{align}

    By choosing $r=c_2/\varepsilon$ for sufficiently large $c_2$ so that $r'\ge 16s_0g/\varepsilon$, we can make the expected patching cost in Eq.~\eqref{eq:patchingcost} at most $\frac{\varepsilon}{4}|e_0|$. Similarly, by choosing $c_1$ sufficiently large, we can make the expected rerouting cost in Eq.~\eqref{eq:reroutingcost} at most $\frac{\varepsilon}{4}|e_0|$. Therefore,
    \begin{equation}
        \E_{(a,b)}[|e_2| - |e_0|] \le \frac{\varepsilon}{2}|e_0|.
    \end{equation}
    By Markov's inequality, $|e_2|\le (1+\varepsilon)|e_0|$ with probability at least 1/2.
\end{proof}

\section{Completing the proof for 2D TTI codes}
\label{sec:generalTTI}
In this section, we extend our results to all 2D TTI codes, proving Theorem~\ref{thm:PTAS2DTTI}.
Given the results of the previous sections, the remaining part of the proof proceeds in two steps. First, we show that it suffices to consider syndrome configurations where there are no excitations on the plaquettes next to the boundary of any square. This is because we can move any excitation away from the boundary of a square. We show that the cost incurred by such operations is minimal relative to the optimal error weight when $s_0 =\Omega(1/\varepsilon)$ and we randomize over shifts $(c,d)$. Then, we show that for any such syndrome, the Patching and Rerouting Properties~\ref{lem:patching}~and~\ref{lem:rerouting} hold. This implies the Structure Theorem~\ref{thm:structure}, and consequently, the dynamic program in Algorithm~\ref{alg:DP} finds a $(1+\varepsilon)$-approximation of the minimum-weight error by Theorem~\ref{thm:DP}.

\begin{lemma}
    \label{lem:randomshiftmovesyndrome}
    Let $\mc C$ be a 2D TTI code such that any qubit is in support of at most $w$ stabilizer generators. Let $0 < \varepsilon < 1$ and $s_0\ge 192qw/\varepsilon$.
    Then for any syndrome $\sigma$ of $\mc C$, the following is true with probability at least $1/2$ over a random shift $(c,d)$. There is an operator $e'$ with syndrome $\sigma'$ such that $\tilde\sigma := \sigma + \sigma'$ has no excitations on plaquettes adjacent to the boundary of any shifted square, and if $\tilde e_\varepsilon$ is a $(1+\varepsilon/4)$-approximation to the minimum-weight error for syndrome $\tilde\sigma$, then $\tilde e_\varepsilon + e'$ is a $(1+\varepsilon)$-approximation to the minimum-weight error for syndrome $\sigma$.
\end{lemma}

\begin{proof}
    The error $e'$ consists of operators that move excitations that are adjacent to a line one plaquette away. Each excitation $\sigma_i$ contributes up to $C:=6q$ to the weight of $e'$ if it is adjacent to a vertical line or a horizontal line (or $2C$ if it is adjacent to both). Over random shifts $(c,d)$, the probability that $\sigma_i$ is next to a vertical line is $2/s_0$, as is the probability of it being next to a horizontal line. Therefore, the expected weight of $e'$ is
    \begin{align}
        \E_{(c,d)}[|e'|] &\le \sum_i [\Pr(\sigma_i \text{ adjacent to a horizontal line}) + \Pr(\sigma_i \text{ adjacent to a vertical line})]C\\
        &= \frac{4}{s_0}|\sigma|C
        \le \frac{4wC}{s_0}|e|
        \le \frac{\varepsilon}{8}|e|.
        \label{eq_1}
    \end{align}
    The first and second inequalities in Eq.~\eqref{eq_1} follow from, respectively, the definition of $w$ and the assumption on $s_0$.
    By Markov's inequality, $|e'|<\frac{\varepsilon}{4}|e|$ with probability at least $1/2$.

    Let $e_0$ and $\tilde e_0$ be minimum-weight solutions for $\sigma$ and $\tilde\sigma$, respectively. Then $e_\varepsilon + e'$ is an error with syndrome $\sigma$, and with probability at least $1/2$, its weight is bounded as follows
    \begin{align}
    \label{eq_2}
        |e_\varepsilon + e'| &\le |e_\varepsilon| + |e'|
        \le \left(1 + \frac{\varepsilon}{4}\right)|\tilde e_0| + \frac{\varepsilon}{4} |e_0|
        \le \left(1 + \frac{\varepsilon}{4}\right)|e_0 + e'| + \frac{\varepsilon}{4} |e_0|\\
        &\le \left(1 + \frac{\varepsilon}{4}\right)\left(|e_0| + \frac{\varepsilon}{4} |e_0|\right) + \frac{\varepsilon}{4} |e_0|
        = \left(1 + \frac{3}{4}\varepsilon + \frac{1}{16}\varepsilon^2\right)|e_0|
        \le (1+\varepsilon)|e_0|.
    \end{align}
    The third inequality in Eq.~\eqref{eq_2} holds since $\tilde e_0$  and $e_0+e'$ have syndrome $\tilde\sigma$ and $\tilde e_0$ has minimum weight.
\end{proof}

\begin{lemma}[Patching Property for 2D TTI codes]
    \label{lem:patching_TTI}
    For any 2D TTI code, Property~\ref{lem:patching} holds for syndromes $\sigma$ without any excitations on plaquettes adjacent to the boundary of any shifted square.
\end{lemma}
\begin{proof}
    We prove the lemma for a horizontal line segment $T$, with the case of a vertical line segment similar.
    Define $T^\circ$ to be the interior of $T$, $\tilde T$ to be the portion of $T$ excluding two vertices on each end, and $\tilde T^+$ to be the union of $\tilde T$ and its vertical translations one unit up and down; see Fig.~\ref{fig:TTI_patching}(a). Let $e_g$ be the global error and $e_{\tilde T^+}$ and $e_{(\tilde T^+\cup T^\circ)\setminus \tilde T}$ be the restrictions of $e_g$ to the sets indicated by the subscripts. Our approach is to find an error $e'_{B_T}$ supported on $B_T$ that is equivalent to $\tilde e_{\tilde T^+}$. Then $e':=e'_{B_T} + e_{(\tilde T^+\cup T^\circ)\setminus \tilde T}$ will be equivalent to $e = e_{\tilde T^+} + e_{(\tilde T^+\cup T^\circ)\setminus \tilde T}$, and it is supported on $B_T$ because $\tilde T^+\cup T\subseteq B_T$.
    
    Consider the syndrome of $e_{\tilde T^+}$. It can only have excitations on plaquettes containing vertices in $\tilde T^+$ but not the ones with all four vertices in $\tilde T^+$. This is because any excitation on those plaquettes would also be one produced by $e_g$, which does not have any excitations on plaquettes adjacent to $T$.    
    We denote the set of plaquettes with possible excitations $J$.
    The error $e'_{B_T}$ is defined by moving all excitations through $J$ to a single location; see Fig.~\ref{fig:TTI_patching}(b). Thus, $e'_{B_T}$ only has support on the vertices adjacent to plaquettes in $J$, which is a subset of $B_T$ and has size $4s$.
    The support of $e'$ will also be a subset of these vertices. Hence, the first property holds with $g=4q$.
    The only vertices adjacent to a plaquette in $J$ are the two at the endpoints of $\tilde T$ and the two one unit away. This implies the second property with $\beta=4q$.
    The third property is due to Lemma~\ref{lem:bufferproperties}.
\end{proof}

\begin{figure}[htpb]
    \centering
\includegraphics[width=0.95\linewidth,trim={1cm 1cm 1cm 1cm},clip]{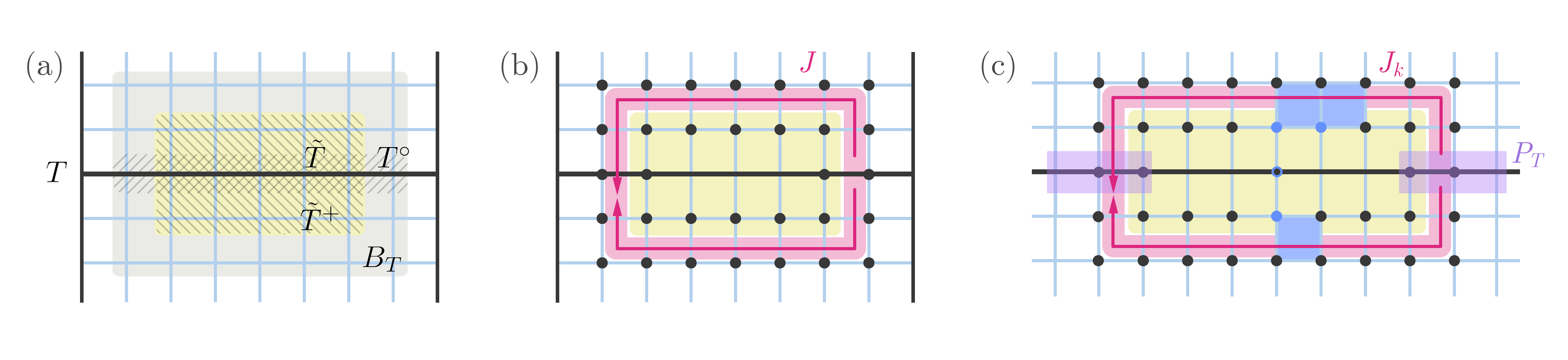}
    \caption{
    (a) The sets of qubits $\tilde T$, $\tilde T^+$, $T^\circ$, $B_T$ used in proving the Patching Property for 2D TTI codes.
    (b) The error $e_{\tilde T^+}$ supported on $\tilde T^+$ (yellow region) causes excitations only within the set of plaquettes $J$ in red (due to the assumption on the global syndrome $\sigma$).
    Moving them along the red path annihilates the excitations and gives an equivalent error $e_{B_T}'$ with constant support on $T$. (c) An error $e_{T_k^+}$ (blue dots) supported on $T_k^+$ (yellow region) creates excitations on the blue plaquettes.
    In general, the syndrome is supported on $J_k$ (red plaquettes).
    Moving them along the red path annihilates the excitations and gives an equivalent error $e_k'$ whose support on $T$ is contained in the portals $P_T$.
    }
    \label{fig:TTI_patching}
\end{figure}

\begin{lemma}[Rerouting Property for 2D TTI codes]
    \label{lem:rerouting_TTI}
    For any 2D TTI code, Property~\ref{lem:rerouting} holds for syndromes $\sigma$ without any excitations on plaquettes adjacent to the boundary of any shifted square.
\end{lemma}
\begin{proof}
    We again prove the lemma for a horizontal line segment without loss of generality. Let $T^\circ$ be the interior of $T$, $\{T_k\}_k$ be the portions of $T$ between two portals, inclusive of the two portals, and $T_k^+$ be the union of $T_k$ and its vertical translations up and down; see Fig.~\ref{fig:TTI_patching}(c). Let $e_g$ be the global error and $e_{T_k^+}$ denote the restriction of $e_g$ to $T_k^+$. For each $T_k$ that intersects the support of $e$, we find an error $e_k'$ supported on $B_T$ that is equivalent to $e_{T_k^+}$. Let $K$ be the set of all such $k$. Then $e' = \sum_{k\in K} e_k' + \tilde e$ will be equivalent to $e = \sum_{k\in K} e_{T_k^+} + \tilde e$, where $\tilde e$ is the restriction of $e_g$ to the set $(\bigcup_{k\in K} T_k^+\cup T^\circ)\setminus \bigcup_{k\in K} T_k$. Also, $e'$ is supported on $B_T$ because each summand is supported on $B_T$.

    We annihilate the excitations caused by $e_{T_k^+}$ similarly as in the proof of Lemma~\ref{lem:patching_TTI} and as illustrated in Fig.~\ref{fig:TTI_patching}(c). The set $J_k$ consisting of the plaquettes with a nonzero proper subset of its vertices in $T_k^+$ contains the possible excitations. The error $e_k'$ is defined by moving all excitations through $J$ to a single location, which has support contained in $B_T$. For $T_k$ on the ends of the $T$, this uses the fact that we placed an extra portal distance two away from each endpoint of $T$. Note that the support of $e_k'$ on $T$ is contained in the portals. Because $\tilde e$ does not have support on $T\setminus P_T$, the first property is satisfied. For each $k\in K$, the error $e$ has support on $T_k$, and we add support of $e'$ on up to four vertices in $T$. This shows the second property with $\beta'=4q$. From our placement of portals, each $T_k$ has length $s_k \le \frac{L}{2^i(m'-1)} - 2$. For each $k\in K$, the equivalent error $e_k'$ may contribute support on up to
    \begin{equation}
        4(s_k+3)+4 \le 6(s_k+2) \le \frac{6L}{2^i(m'-1)} \le \frac{24qL}{2^im}
    \end{equation}
    vertices by the geometry of $J_k$. The first inequality uses the fact that $s_k\ge 2$ and the last inequality assumes that $m'\ge 2$ so that $m = (3m'-2)q\le 4(m'-1)q$. Thus, the error weight may be increased by a factor of at most $\frac{g'L}{2^i m}$ with $g'=24q^2$, implying the third property.
\end{proof}

We are now ready to prove the main theorem.

\begin{proof}[Proof of Theorem~\ref{thm:PTAS2DTTI}]
    Let $s_0\ge 192qw/\varepsilon$ be a power of $2$, $m=4c_1(\log L)/\varepsilon$, and $r=4c_2/\varepsilon$ with $c_1$ and $c_2$ as in Theorem~\ref{thm:structure}. In our randomized algorithm, the integers $a,b,c,d$ are picked uniformly randomly in the ranges $0\le a,b < L/s_0$ and $0\le c,d < s_0$. We compute $e'$ and $\tilde\sigma$ from Lemma~\ref{lem:randomshiftmovesyndrome}, which can be done with time complexity $O(L^2)$ because $e'$ consists of operators that move excitations one unit away from lines and $\tilde\sigma$ is the resulting syndrome. We run Algorithm~\ref{alg:DP}, which has time complexity $L^2(\log L)^{O(1/\varepsilon)}$ by Theorem~\ref{thm:DP}, to obtain $\tilde e_\varepsilon$. We then output $\tilde e_\varepsilon + e'$. This algorithm has the required time complexity.

    There is at least a $1/2$ probability over the choice of $(c,d)$ that the conclusion of Lemma~\ref{lem:randomshiftmovesyndrome} is true. In particular, $\tilde\sigma$ has no excitations on plaquettes adjacent to the boundary of any shifted square, so Lemmas~\ref{lem:patching_TTI} and~\ref{lem:rerouting_TTI} imply the Patching and Rerouting Properties~\ref{lem:patching}~and~\ref{lem:rerouting}. By the Structure Theorem~\ref{thm:structure} (replacing $\varepsilon$ with $\varepsilon/4$), there is at least a $1/2$ probability over the choice of $(a,b)$ that an $r$-light error with approximation ratio $(1+\varepsilon/4)$ exists for the syndrome $\tilde\sigma$. Algorithm~\ref{alg:DP} finds the optimal $r$-light error by Theorem~\ref{thm:DP}, so $\tilde e_\varepsilon$ is also a $(1+\varepsilon/4)$-approximation for the syndrome $\tilde\sigma$. Lemma~\ref{lem:randomshiftmovesyndrome} therefore implies that $\tilde e_\varepsilon + e'$ is a $(1+\varepsilon)$-approximation for the original syndrome $\sigma$. The combined success probability is at least $1/4$.

    To derandomize the algorithm, we can simply iterate over all $a$, $b$, $c$, $d$ and pick the output of minimum weight. This incurs an extra $L^2$ factor in time complexity since there are $L^2$ choices of $(a,b,c,d)$. The space complexity is increased by $O(L^2)$ to store the best solution so far, which is negligible compared to the trivial bound of $L^2(\log L)^{O(1/\varepsilon)}$ for each iteration.
\end{proof}

\section{Further extensions}

We consider several extensions of our results. In Sec.~\ref{subsec:removeassumptions}, we discuss how to remove some of the simplifying assumptions in our 2D models, such as the lattice size and the boundary conditions. In Sec.~\ref{subsec:higherdimension}, we explain generalizations of our techniques to higher-dimensional topological codes and quantum memories, including the toric code with phenomenological noise.

\subsection{Removing the simplifying assumptions on 2D TTI}
\label{subsec:removeassumptions}

Let us comment on the setting in which $L$ is not a power of two. We define the ``squares'' to be slightly rectangular by rounding lines to the nearest integer coordinate. That is, for a shift $(c,d)$, we can define the level-$i$ ``squares'' to be partitioned by the lines $x=c+\lfloor kL/2^i\rfloor$ and $y=d+\lfloor kL/2^i\rfloor$ for $0\le k < 2^i$.
A shift by $(a,b)$ is a ``translation'' by approximately $aL/2^{i_0}$ units right and $bL/2^{i_0}$ units up so that the level-$i_0$ squares are bijectively mapped to each other.
Although the ``squares'' are deformed in this mapping due to their nonuniform dimensions, this does not affect the analysis.
Portals can also be defined by rounding the coordinates appropriately so that the portals of the larger ``squares'' are also portals of the smaller ``squares''.
With these modifications, the proof is essentially unchanged---the dynamic program still finds the optimal $r$-light error, which is a $(1+\varepsilon)$-approximation for the minimum-weight decoding problem. These modifications also work if the lattice itself is not square, e.g., having size $L_1\times L_2$ for comparable $L_1$ and $L_2$.

Our results also hold for topological codes with open boundary conditions. 
Here, we still shift the squares in a periodic way, which means that some squares may be ``cut off'' by the boundary and consist of two (or four) components placed along opposite sides of the lattice.
Since we never used any assumption about squares being connected, all the proofs are unchanged. For example, the dynamic program still searches for the solutions of the decoding subproblems, which may consist of errors on the boundary of the lattice if a square is cut off. In the Patching and Rerouting Properties, if a line segment is cut off by the boundary of the lattice, the total charge caused by the error in each connected component is neutral, so the excitations within each connected component can be annihilated independently.

We also remark that the direct translation of Arora's PTAS for Euclidean problems~\cite{AroraTSP} would dissect squares only until they contain a constant number of excitations, resulting in a data structure called a quadtree. However, in our setting, it is not obvious how to efficiently solve the decoding subproblem directly on a large square of non-constant size, even when there is a constant number of excitations inside.
This is why we dissect squares until the highest-level ones are of constant size, where exhaustive search is efficient. If the base cases of the dynamic program with larger squares were efficiently solvable, we could potentially improve the time complexity of our algorithm from $L^2(\log L)^{O(1/\varepsilon)}$ to $|\sigma|(\log L)^{O(1/\varepsilon)}$.

\subsection{Higher-dimensional topological codes and quantum memories}
\label{subsec:higherdimension}

\begin{figure}[htpb]
    \centering
    \includegraphics[width=0.8\linewidth]{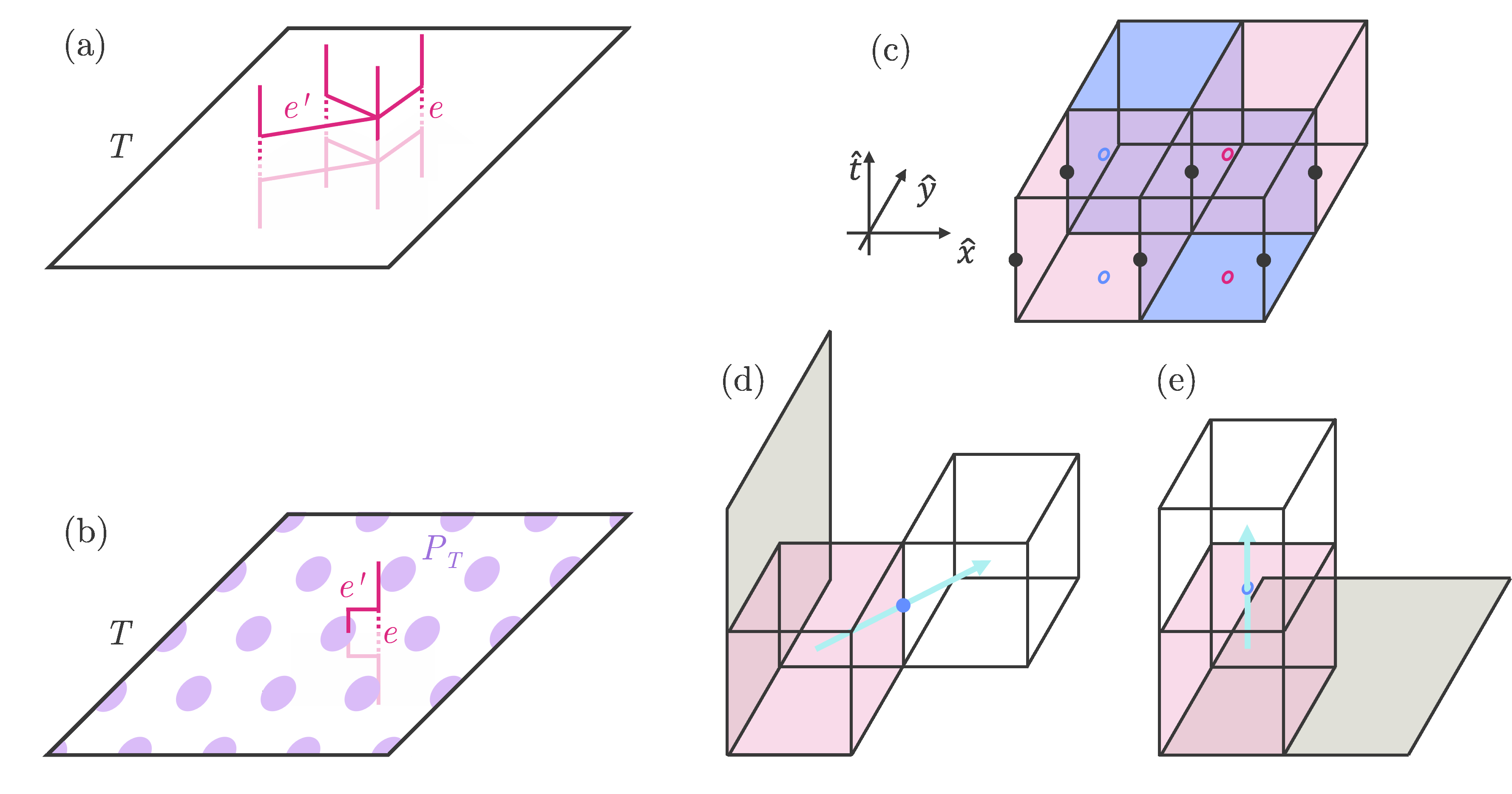}
    \caption{(a) The Patching Property for higher-dimensional decoding problems. An error $e$ (red vertical lines) is replaced by an equivalent one $e'$ (solid red lines) whose support on the $(D-1)$-dimensional face $T$ has constant size. The paths taken by $e'$ on either side of $T$ are spanning trees.
    (b) In the Rerouting Property, an error $e$ (red vertical line) that passes through a $(D-1)$-dimensional face $T$ is replaced by an equivalent error (solid red line) whose support on $T$ is contained in the portals $P_T$. 
    (c) $X$ and $Z$ detectors in the quantum memory setting with the $(2+1)$D toric code form a checkerboard pattern in space and are repeated in the time direction $\hat t$.
    We depict them as red and blue cubes, respectively.
    Each detector comprises two measurement errors on horizontal faces (hollow dots) and four qubits on vertical edges (solid dots).
    We can move an $X$ excitation (shaded cube) away from either (d) a vertical plane using a $Z$ qubit error or (e) a horizontal plane using an $X$ measurement error.
    }
    \label{fig:higherdimension}
\end{figure}

Our approach to the decoding problem leverages the presence of point-like excitations connected by string-like errors.
This structure is not specific to two dimensions---certain higher dimensional topological codes and quantum memories exhibit it as well.
Concretely, the following three decoding scenarios have precisely this structure: 
(i) the 3D subsystem toric code~\cite{Kubica2022, Bridgeman2024} and the 3D gauge color code~\cite{Bombn2015gauge, Kubica2015} with Pauli $X$, $Y$ and $Z$ errors, 
(ii) a variant of the stabilizer color code
in $D\geq 3$ that has point-like excitations for Pauli $Z$ errors\footnote{
For the stabilizer toric code in $D\geq 3$, there is likewise a variant with point-like excitations for Pauli $Z$ errors.
Analogously as for the 2D toric code, the minimum-weight perfect matching decoder finds a minimum-weight recovery in this case, so minimum-weight decoding for the higher-dimensional toric code with Pauli $Z$ errors is efficiently solvable.} 
(for Pauli $X$ errors, by contrast, the corresponding excitations are no longer point-like and instead correspond to higher-dimensional objects)~\cite{Bombin2007,Kubica_thesis},
and (iii) quantum memory with the $(2+1)$D toric code and phenomenological noise comprising Pauli $X$, $Y$ and $Z$ errors together with measurement errors~\cite{Dennis2002}
(the scenario with circuit-level noise is analogous, albeit with a more complex decoding graph).

In all three cases, point-like excitations can branch and, consequently, following similar arguments as in Ref.~\cite{gu2026colorcodesurfacecode}, one can establish NP-hardness of the minimum-weight decoding problem via a reduction from the three-dimensional matching problem, which is NP-complete~\cite{Karp1972,GareyJohnson90NPcompleteness}; nevertheless, because of the decoding structure, our PTAS extends to them.
In the remainder of this subsection, we sketch the high-level idea of how to modify our PTAS for a $D$-dimensional decoding problem, similarly to how Arora's PTAS for Euclidean problems is generalized to higher dimensions~\cite{AroraTSP}.

In what follows, we assume that $D$ is a constant and thus ignore prefactors that are only dependent on $D$.
The lattice will be recursively dissected into $D$-dimensional hypercubes instead of squares, and shifted dissections are defined in a similar way with $D$ shift directions instead of two.
We place portals on the $(D-1)$-dimensional faces of these hypercubes, 
which consist of a regularly spaced subgrid of vertices on the face as well as vertices within some constant distance of the subgrid; see Fig.~\ref{fig:higherdimension}(b).
The number of portals on each $(D-1)$-dimensional face is chosen to scale as
$m = O\left((\frac{1}{\varepsilon}\log L)^{D-1}\right)$, and we are interested in the $r$-light solutions that use $r=O\left((\frac{1}{\varepsilon})^{D-1}\right)$ of them.

In the dynamic program, we find the optimal $r$-light solution on a given shifted dissection. There are $O(L^D)$ hypercubes $S$ and $O(m^{2Dr})$ decoding subproblems on each hypercube from choosing an error $f$ supported on $\partial S$ using at most $r$ portals on each of its $2D$ faces of dimension $(D-1)$.
The decoding subproblems are solved in the same way as before.
The base cases on the highest-level hypercubes are found through exhaustive search.
Each decoding subproblem on a lower-level hypercube $S$ with error $f$ is solved by iterating over the $r$-light lifts $f^+$ supported on the boundaries of all $2^D$ of its subhypercubes, piecing together the solutions for the subhypercubes, and picking a minimum-weight error configuration.
There are $D2^{D-1}$ faces of these subhypercubes that are ``internal'' with respect to the hypercube, so there are $O(m^{D2^{D-1}r})$ many $r$-light lifts $f^+$. For each $f^+$, a constant number $2^D$ of solutions need to be looked up.
Therefore, the total time complexity is $O(L^Dm^{2Dr}m^{D2^{D-1}r})=L^D(\log L)^{O\left((1/\varepsilon)^{D-1}\right)}$ (with trivially the same bound for the space complexity).

To show correctness of the algorithm, we need to prove the Structure Theorem stating that with probability at least $1/2$ over random shifts, there is an $r$-light lift with weight within a multiplicative factor of $(1+\varepsilon)$ of the minimum for the syndrome $\sigma$.
By applying (modified versions of) the Patching and Rerouting Properties, the minimum-weight error is transformed to an $r$-light one.
These properties replace an error supported on the face $T$ of a hypercube with an equivalent one supported on its buffer $B_T$ consisting of all translations of $T$ in the orthogonal direction within some constant distance.
Note that the transformations are done only for the purposes of the proof, as the dynamic program will find the optimal $r$-light solution.

In $D$ dimensions, the Patching Property states that for any weight-$k$ error $e$ supported on the interior of a $(D-1)$-dimensional face $T$ with side length $s$, there is an equivalent error $e'$ of weight $O(k^{1-(1/(D-1))}s)$ supported on $B_T$ that intersects $T$ at most a constant number of times.
Note that the bound on $e'$ now depends on $k$ for $D>2$.
This property can be proven in a similar way as in the 2D case by moving all excitations caused by $e$ away from $T$ and routing them to the same location through a path that crosses $T$ at a single location. This annihilates all the excitations because the total charge of all excitations caused by $e$ is neutral.
A well-known result from graph theory states that a minimum-weight spanning tree in $D-1$ dimensions has total length $O(k^{1-(1/(D-1))}s)$~\cite{SteeleSnyder1989};
one can choose to move the excitations on either side of $T$ along such a spanning tree.
An example of applying the Patching Property is shown in Fig.~\ref{fig:higherdimension}(a).
We may also use the Patching Property for lower-dimensional faces of a hypercube to clean the error from those boundaries.

For the Rerouting Property, we deform an error $e$ supported on a $(D-1)$-dimensional face $T$ of a level-$i$ hypercube so that it only passes through the portals in $T$, as shown in Fig.~\ref{fig:higherdimension}(b). This is done by moving the excitations away from $T$ and rerouting them through the closest portal. Because the portals form a subgrid of $T$ with spacing $\Theta(m^{-1/(D-1)}L/2^i)$, the equivalent error $e'$ will have weight at most $|e'| = O(m^{-1/(D-1)}L/2^i)|e|$, and its support on $P_T$ will only increase by a constant factor.

The Structure Theorem can then be proven in a similar way as in the 2D case.
By randomizing over the shifts, the probability that a given hyperplane has level $i$ is at most $2^i/L$, and we can show that $\E[|e'|-|e|]\le \frac{\varepsilon}{2}|e|$. This means that $|e'|\le (1+\varepsilon)|e|$ with probability at least $1/2$ by Markov's inequality. We skip over the details of these calculations because it uses the same ideas as in the proof of Theorem~\ref{thm:structure} and is essentially unchanged from the presentation in Ref.~\cite{AroraTSP} for Euclidean problems.

However, there is one important subtlety in proving the Patching and Rerouting Properties in $D>2$ dimensions.
For a general topological code, moving an excitation on a hypercube away from its faces may require acting on qubits on these faces.
For example, if the lattice is sufficiently coarse-grained, a stabilizer on a hypercube may only act on qubits at a single vertex. This is why for general 2D TTI codes, we prove the property for syndromes without excitations on plaquettes next to any line. In the proof of the Patching Property, we work with an error $\tilde e_{\tilde T^+}$ that is the restriction of the global error to a region around $T$, and its excitations on plaquettes adjacent to $T$ can only be at the ends of $\tilde T^+$. This means that we only have to move a constant number of excitations away from plaquettes adjacent to $T$, and $e'$ will have constant intersection with $T$. In higher dimensions, the region to $\tilde T^+$ would be at least two-dimensional and have a boundary of growing size, which means that we might not be able to guarantee constant support of $e'$ on $T$. A similar issue arises for the Rerouting Property. We expect this caveat to be an artifact of this particular proof technique and that an $r$-light error should approximate the minimum-weight solution in general.

For the three decoding scenarios mentioned at the beginning of this subsection, it is possible to move excitations away from a plane without using qubits on the plane.
For example, in the quantum memory setting with the (2+1)D toric code and phenomenological noise, we may consider the rotated representation with $X$ and $Z$ detectors forming cubes as shown in Fig.~\ref{fig:higherdimension}(c).
Recall that in this setting, a detector is obtained by combining measurement outcomes of the same stabilizer generator from two consecutive time steps; we consider flipped detectors as locations of point-like excitations.
Similarly as in the two-dimensional case, point-like excitations are connected by string-like errors comprising Pauli $X$, $Y$ and $Z$ operators, as well as measurement errors.
Here, an excitation on a cube may be caused by either a measurement error associated with one of the two faces above and below the cube or a qubit error associated with one of the four vertical edges of the cube.
Therefore, the excitation can be moved away 
from a vertical plane using a qubit error on either of the two vertical edges on the other side of the plane (Fig.~\ref{fig:higherdimension}(d)); it can be moved away 
from a horizontal plane using a measurement error on the opposite face of the plane (Fig.~\ref{fig:higherdimension}(e)).
Once the excitation is away from a plane, it can then be moved freely using string-like operators to satisfy the Patching and Rerouting Properties.
Our analysis carries over in the three decoding scenarios, showing that a PTAS exists for these minimum-weight decoding problems.

\section{Discussion}

In our work, we established a PTAS for minimum-weight decoding of topological codes in two dimensions and beyond---to our knowledge, the first PTAS for this problem.
The key structure we exploited is that of point-like excitations connected by string-like operators.
Our results complement previous works establishing the NP-hardness of the exact problem~\cite{Hsieh2011,Kuo2020,walters2026CCNPhard,gu2026colorcodesurfacecode}, i.e., while finding a minimum-weight recovery is intractable, for any $\varepsilon > 0$ we can find in polynomial time a recovery operator whose weight is within a multiplicative factor of $1 + \varepsilon$ of the minimum weight.

As an immediate consequence, our PTAS yields an explicit, computationally efficient decoder for the 2D color code that corrects any error of weight up to $(1-\varepsilon)d/2$, where $d$ is the code distance, approaching the optimal half-distance bound.
This improves on all previous color-code decoders, which either guarantee correcting only a smaller fraction of the distance~\cite{Delfosse2014,Kubica2023CCrestrictiondecoder,Chamberland2020,Sahay2022} or come with no provable guarantees~\cite{Lee2025}.

We also expect our approach to be of practical value.
Many of the best-performing decoders seek low-weight recovery operators consistent with the observed syndrome; the prototypical example is the minimum-weight perfect matching decoder for the toric code in the memory setting with the circuit-level noise~\cite{Dennis2002,fowler2012}.
These decoders, however, do not necessarily find a minimum-weight recovery, which can be an intractable task.
Since our PTAS serves as a computationally efficient decoder that approximates minimum-weight error, we expect it to be competitive with existing decoders; testing it numerically would be an exciting next step.
On the theoretical side, it would be interesting to prove that approximate minimum-weight decoders have a non-zero QEC threshold and to understand how that threshold may depend on $\varepsilon$.

\begin{acknowledgments}
S.G. and A.K. acknowledge support from the NSF (QLCI, Award No. OMA-2120757), IARPA and the Army Research Office (ELQ Program, Cooperative Agreement No. W911NF-23-2-0219).
L.W. is supported by NSERC Fellowship PGSD3-587672-2024 and thanks the Yale Quantum Institute for hosting
her during the project.

\emph{Note added.---}We would like to bring the reader's attention to independent and concurrent work by Mark Walters~\cite{Walters26}, which provides a PTAS for minimum-weight decoding in the 2D color code.

\end{acknowledgments}

\bibliography{ref}
\end{document}